\documentclass[12pt,a4paper]{article}
\usepackage[english]{babel}
\usepackage[latin1]{inputenc}
\usepackage{amsfonts, amssymb, mathrsfs, amsmath, amsthm}
\usepackage{graphicx,psfrag,epsf}
\usepackage{xcolor}
\usepackage{geometry}
\usepackage{calc}
\usepackage{tikz}
\usepackage{subfig}
\usepackage{enumerate}
\usepackage{url} 
\usepackage{stackrel}

\title{Confidence regions for univariate and multivariate data using permutation tests}
\author{Niels Lundtorp Olsen{$^\dagger$} \\
	Department of Applied Mathematics and Computer Science \\
Technical University of Denmark}
\date{Jun 2022}

\newcommand{\eb}{\Leftrightarrow}
\newcommand{\pil}{\rightarrow}

\newcommand{\M}[1]{\mathbb{#1}}
\newcommand{\R}{\mathbb{R}}

\newcommand{\prikt}[1]{#1 \dots #1}

\newcommand{\E}{\mathrm{E}}

\DeclareMathOperator*\uaf{\perp\mskip-11mu\perp}

\newcommand{\brak}[1]{\langle #1 \rangle}

\newtheorem{prop}{Proposition}
\newtheorem{lemma}{Lemma}

\newtheorem{eks}{Example}
\theoremstyle{definition}
\newtheorem{defi}{Definition}

\begin{document}
	
\maketitle

\begin{abstract}
	Confidence intervals are central to statistical inference as a tool to evaluate the type I error risk at a given significance level. 
	We devise a method to construct non-parametric confidence intervals using a single run of a permutation test. This methodology is extended to a multivariate setting, where we are able to handle multiple testing under arbitrary dependence. We demonstrate the method on a weather data set and in a simulation example.
\end{abstract}

\medskip

\noindent%
{\it Keywords:} confidence intervals, permutation tests, multiple testing, non-parametric inference
\medskip

{$\dagger$: nalo@dtu.dk}

\section{Background}
There is a well-known duality between confidence intervals and tests: let $\theta$ be a quantity of interest to be estimated -- if $ \theta_0 \notin CI_\alpha(X)$ on a given significance level $\alpha$, then  ($H_0: \theta = \theta_0$) is rejected, and this has probability $\alpha$ under $H_0$ (at least ideally).
The statistical inference usually goes from having a confidence interval to rejecting/accepting hypotheses, but the other way is also possible (yet rarely done).

There exists a vast literature on hypothesis testing, partly arising from the fact that closed-form solutions are generally not available outside of the linear normal model.  

Confidence intervals are often constructed using asymptotical properties of estimators.
This usually amounts to $\hat{\theta} \pm 1.96 \cdot \hat{\sigma}_{\theta}$, where $\hat{\theta}$ and $\hat{\sigma}_{\theta}$ are the estimate and estimated standard error, respectively.
 However, this approximation becomes increasingly problematic for small sample sizes.

An alternative to parametric models is to use  non-parametric tools, for which the most versatile tool is \emph{permutation testing} (Section \ref{perm-test-sect}). Permutation tests are broadly applicable and require only few assumptions. Permutation tests work well for high-dimensional data and do not require assumptions on the dependence structure.

\paragraph{Multiple testing}
When considering several parameters or hypotheses, multiple testing becomes an issue. Many methods and error quantities have been proposed, we here focus on the \emph{family-wise error rate} (FWER), ie the chance of committing at least one type I error. In terms of multiple confidence intervals, this translates into $\theta$ not belonging to the cartesian product of the marginal confidence intervals. Whereas a large literature exists for tests (and multiple testing) for high-dimensional data, these methods do not straightforwardly convert into confidence intervals.  

Having multiple tests increases the chances of a type I error. There are two closely related issues:
\begin{enumerate}
	\item When having a set of multiple confidence intervals, what is the joint confidence level (ie. the confidence level of the cartesian product)?
	\item How do we construct (or adjust) confidence intervals, such that their joint confidence level is $1-\alpha$, for a given $\alpha$?
\end{enumerate}

The oldest correction method for multiple testing is the \emph{Bonferroni correction}, presented for confidence intervals by \cite {dunn1961}.
The Bonferroni inequality says that if each of $K$ statistical tests/confidence intervals has a type I error chance at most $\alpha$, then the joint statistical test/confidence region has a FWER at most $K \alpha$. Conversely, if we construct an $(1 - \frac{\alpha}{K})$ confidence interval for each parameter, the joint confidence level is at least $1 - \alpha$.

The Bonferroni correction represents the extreme case of type I errors never happening concurrently; another relevant case is independence of the type I errors. Under this assumption, the joint  confidence level of $K$ $(1 - \alpha)$ confidence intervals is $(1 - \alpha)^K$. This is known as the \emph{Sidak correction}. Sidak showed that this adjustment remained valid for arbitrary dependences in the multivariate normal distributions, when constructing confidence intervals for the means \cite{sidak1967}.

A crucial issue is that of \emph{dependence} between the hypothesis tests. If two variables are positively correlated, then the chances of a type I error is also positively correlated (at least when using common methods). This implies that p-values and confidence intervals need less adjustment compared to the independence case.

\begin{figure}[!htb]
	\includegraphics[width=0.6\textwidth]{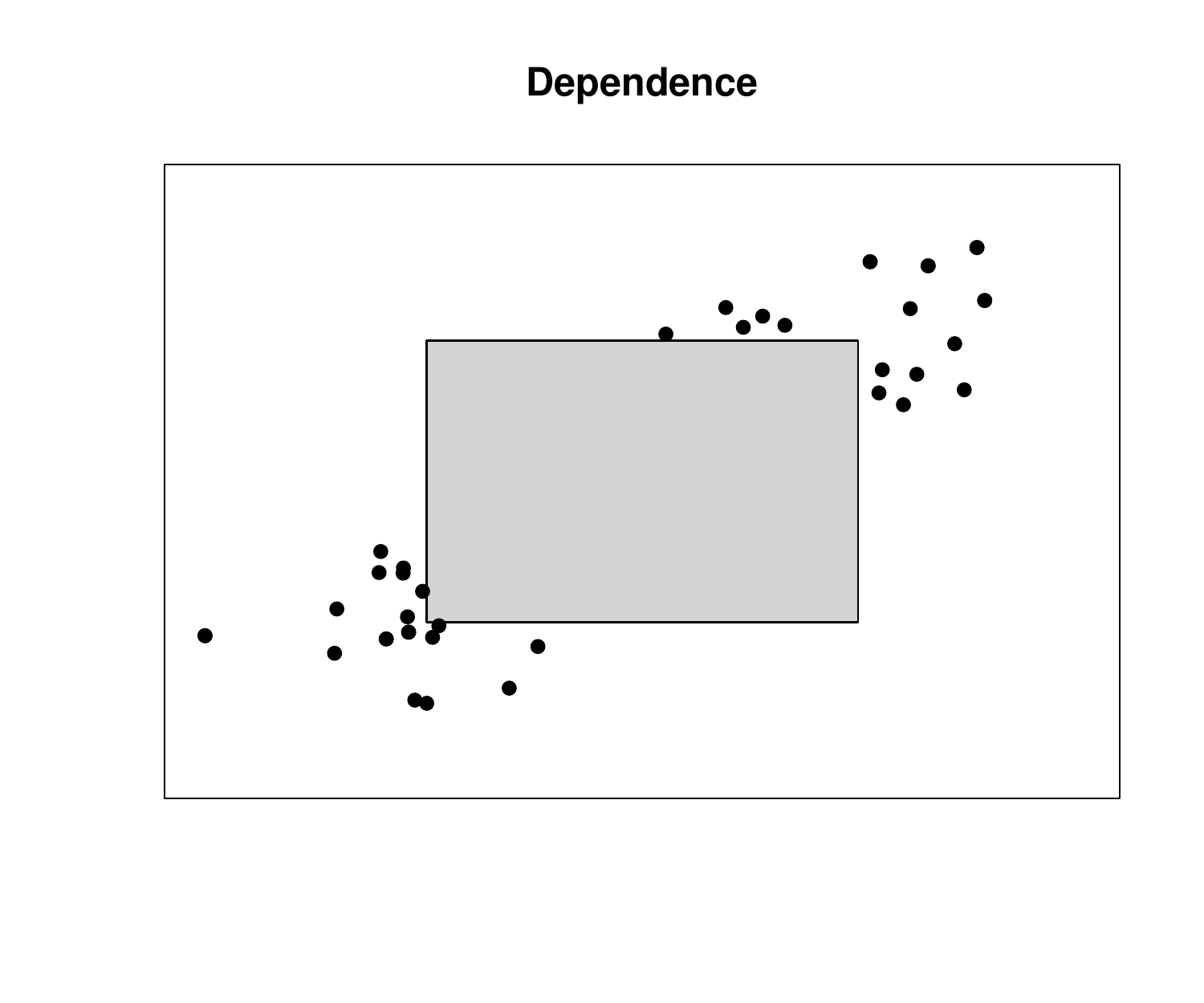}
	\includegraphics[width=0.6\textwidth]{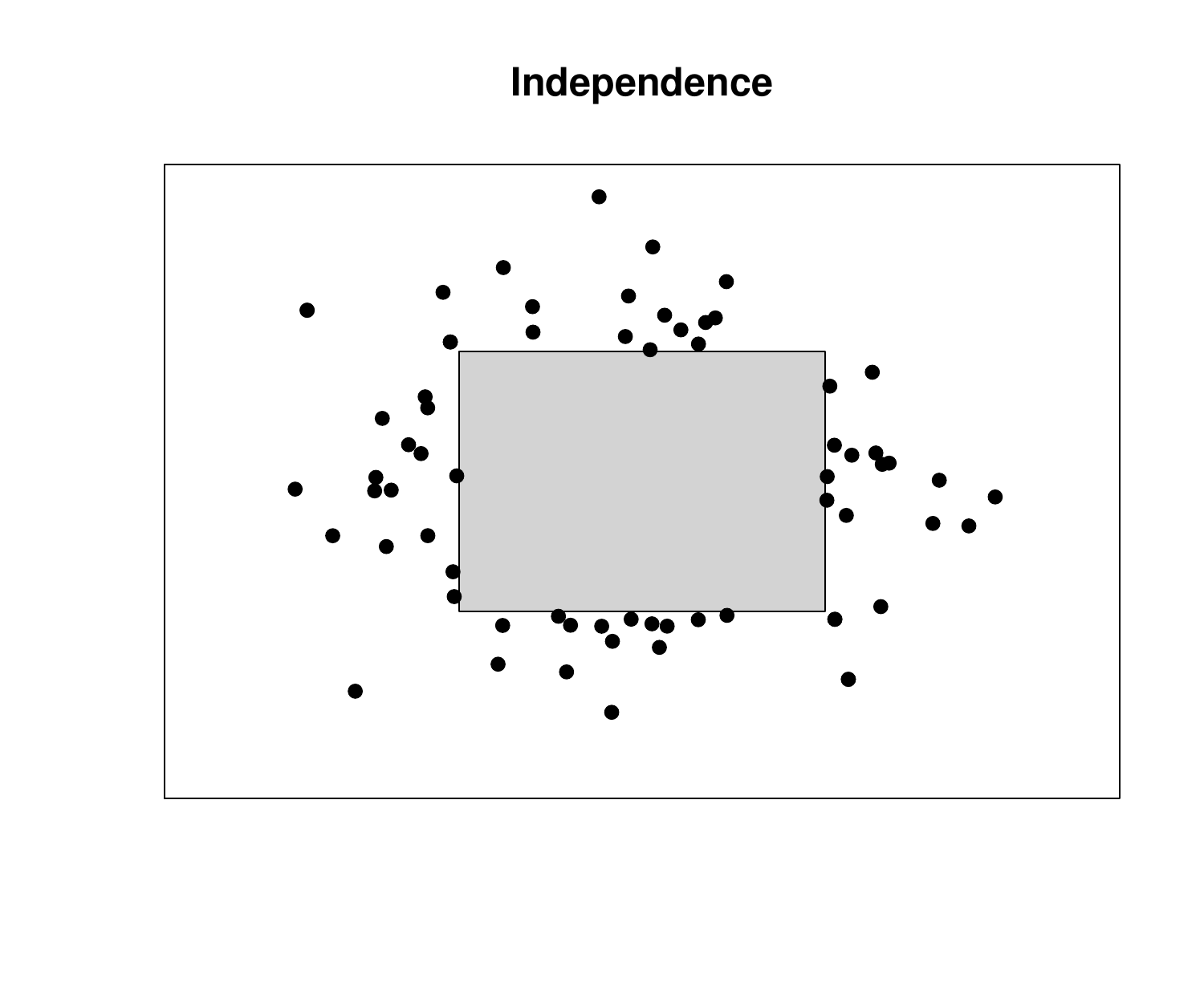}
	\includegraphics[width=0.6\textwidth]{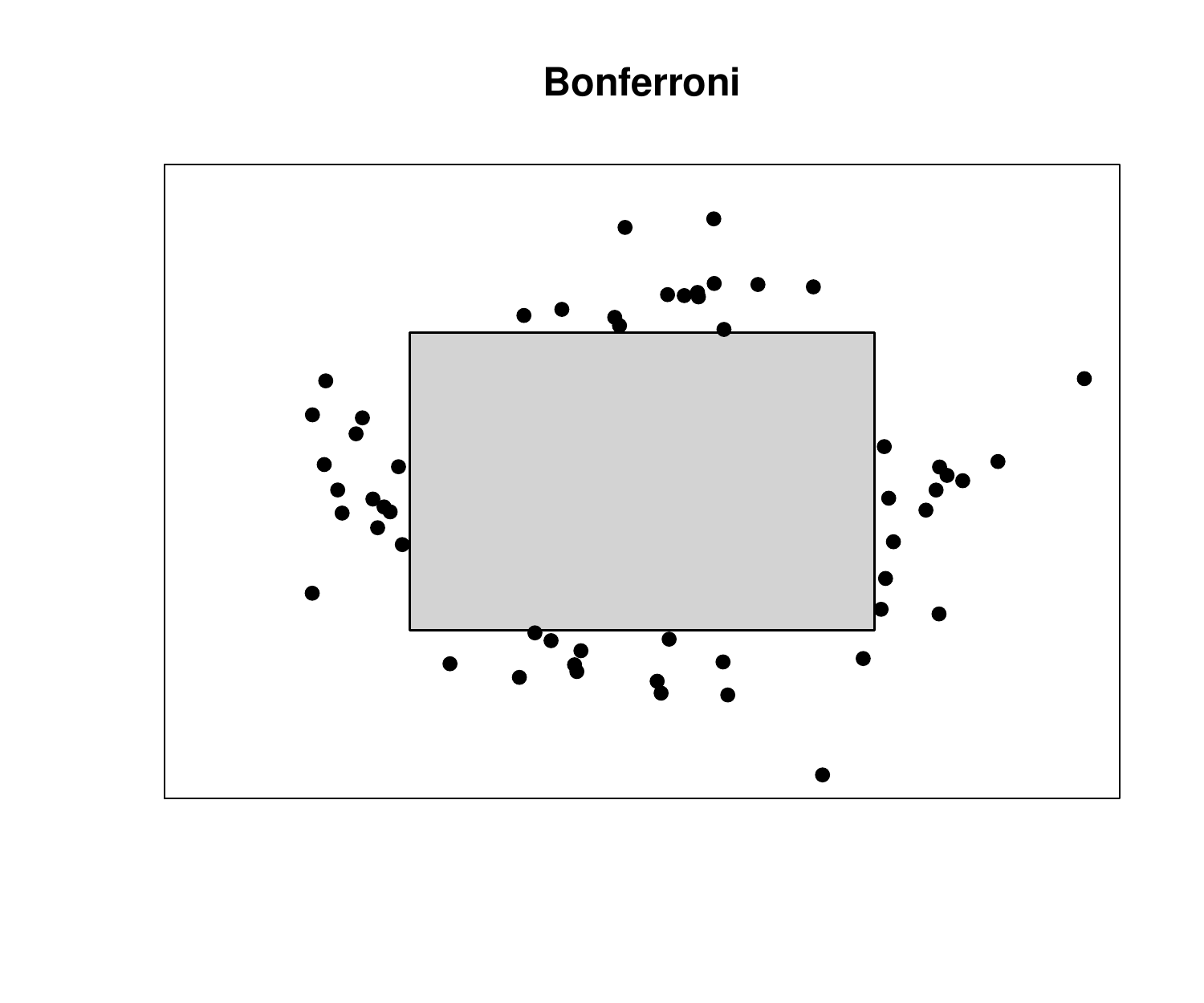}
	\caption{The family-wise error rate in two dimensions. The shaded rectangle represent the product of two $(1-\alpha)$ confidence intervals. The dots represent type I errors in three scenarios. Upper: positive dependence. Middle: independence. Bottom: mutually exclusive.} \label{figtwodim}
\end{figure}

We have illustrated this for the two-dimensional case in Figure \ref{figtwodim}. Here, the shaded region represents a confidence region for $\mu_x$ and $\mu_y$, each on level $(1- \alpha)$, and black dots are estimates outside of this, ie. type I errors. 
Thus, if the type I error probability is $\alpha$, the FWER is given by
$$
2\alpha - P(\text{type I error for }\mu_x \text{ and type I error for }\mu_y),
$$
which in the figure are the shaded regions "across the corners".  

In the case of strong positive correlation, a large fraction of type I errors are both for $\mu_x$ and $\mu_y$, so the FWER is much below $2 \alpha$. 
In the independence case, there is small probability of a joint type II error, so the FWER is slightly below $2 \alpha$. 
In the last case, the two type I errors are mutually exclusive, and the FWER is $2 \alpha$. 
\footnote{We note the slight abuse of the term "confidence". However, if the confidence intervals are on the form (or close to) $\hat{\theta} \pm z_\alpha$ for a fixed $z_\alpha$, this description is valid. }

In summary, there is thus much to be gained, if we are able to correctly assess the "effect" of dependence when constructing or adjusting multiple confidence intervals. 
In particular, this would allow us to adjust "less" in the case of strong dependence, {improving} the statistical inference.

\subsection{Permutation tests} \label{perm-test-sect}
For notation, let  $S_N$ denote the symmetric group of order $N$. We shall identify a permutation $s \in S_N$ with its corresponding permutation function $\R^N \pil \R^N$. We shall use $e$ to refer to the identity permutation. \medskip

Permutation tests are a class of non-parametric tests that tests a hypothesis by \emph{permuting} data $X = (X_1 \prikt{,} X_N) \in \mathcal{X}^N$, using an assumption of exchangeabliity under the hypothesis. 
Permutation tests are commonly used to test {comparisons} such as  two-sample comparison, where $X_1 \prikt{,} X_N$ are iid. under the null hypothesis, but not otherwise.  
One crucial advantage of permutation tests is that $X_i$ can be any kind of data, including multivariate data with a complicated (and unknown) dependence structure, allowing an enormous flexibility and wide scope. %

The main drawback of permutation tests is the computational cost involved. However, with the advances in programming tools and parallel computing, this is a minor issue.
A second drawback is that calculating all $N!$ permutations is unfeasible for all but very small $N$. Therefore, permutation tests are commonly implemented using \emph{Conditional Monte Carlo} (CMC), which uses randomly sampled permutations. This method is well-behaved, but introduces randomness to the result (ie. the $p$-value) due to the random sampling. 
We refer to \cite{pesarin} for a general discussion of permutation tests.

\paragraph{The permutation test}

Let $X = (X_1 \prikt{,} X_N  \in \mathcal{X}^N)$ be a stochastic variable generated by some statistical model $\{P_\theta, \theta \in \Theta\}$. We consider a null hypothesis $H_0 \subseteq \Theta$ such that under $H_0$
$$
(X, \dots, X_N) \stackrel{D}{=} s(X_1, \dots, X_N) , \quad s \in S_N
$$


A permutation test consists of a test statistic $T: \mathcal{X}^N \pil \R$ such that large values of $T(X)$ are evidence against $H_0$\footnote{For simplicity, we only consider one-sided test statistics, one can use two-sided test statistics, too}, and where the distribution of $T$ is invariant to permutations of $X$ under $H_0$.

Let $s_1, \dots s_M$ be random permutations from $S_N$. Define $T_\text{obs} = T(X)$ and $T_m = T \circ s_m(X)$, and let $Q_\gamma$ denote the $\gamma$ quantile of a vector; $Q_{\gamma}(a_1, \dots, a_M) = a_{(\lceil M\gamma \rceil)}$. The permutation test goes as follows:

\begin{quotation}
	{\it
		Using significance level $\alpha$, we reject $H_0$ if $T_\text{obs} > Q_\alpha(T_1, \dots, T_M)$.
	}
\end{quotation}
Ideally, one should use all permutations in $S_N$ for constructing $T_1, \dots, T_M$, but this is unfeasible for all but very small $N$. 
The use of random permutations is referred to as \emph{Conditional Monte Carlo} (CMC)

\begin{prop}
Let $\{s_1, \dots , s_M\} = S_N$.
	Under $H_0$	then
	$$
	P(\text{reject }H_0 ) \leq \alpha
	$$
\end{prop}
When $\{s_1, \dots , s_M\}$ are random samples, the above proposition is only an approximation which becomes increasingly good as $M$ goes to infinity. We suggest to use bootstrapping to assess the uncertainty caused by the random sampling (see also Section \ref{sect-uncertain}).

\paragraph{Related research on confidence intervals using permutation tests}
Confidence intervals have been constructed using permutation tests. \cite{pesarin} outlines an algorithm  where hypotheses $H_0: \theta = \theta_0$ are tested on a fine grid, until a threshold $\epsilon$ has been reached \cite[section 3.4]{pesarin}.
The method presented in this paper gives the same result, but uses only a single run of iterations and does not have a grid-related approximation error.

Furthermore, \cite{pesarin} devises a multivariate extension to the univariate algorithm \cite[4.3.5]{pesarin}.
This is an iterative procedure that in practice requires testing on a fine multivariate grid. Additionally, this procedure introduces an implicit ordering of the variables being tested. 
We are not aware of any examples where this algorithm has been applied. 

The multiple testing procedure presented in this paper is different as it directly uses the results from the univariate method and only considers box-shaped confidence regions.

\subsection{Contributions of this paper}
We devise an algorithm for constructing non-parametric confidence intervals using a single set of permutations. This requires only weak assumptions on the test statistic used, and is easily implemented in software. 
Our proposed algorithm is more arguably a p-value correction method, but carries the same aim as textbook confidence intervals: to define a confidence region with a low, pre-defined chance of making a type I error. 
We do not require any parametric assumptions for the statistical model nor rely on asymptotical properties, thus our proposed method is valid in a wide range of scenarios. 

The methodology is extended to the multivariate case under the same assumptions on the test statistic, but arbitrary dependence between coordinates. 
Our proposed method exploits the "dependence effect" of testing via a permutation test by counting instances where there is a family-wise error. Thus in the case of strong dependence, we obtain a much less conservative estimate of the FWER than, say, Sidaks procedure. 
In detail, our multivariate procedure consists of two parts: (1) a calculation of the adjusted confidence level and (2) an adjustment procedure based on said adjusted confidence level. Only box-shaped confidence regions are considered. 

In summary, our contributions are:
\begin{itemize}
	\item A simple and efficient procedure for constructing single-parameter confidence intervals. Furthermore, there are only minimal assumptions on the distribution, and the procedure does not rely on any asymptotics. 
	
	The related method outlined in \cite[section 3.4]{pesarin} also constructs confidence intervals using permutation tests, but does so by testing on a fine grid. Our procedure has the advantage that it only requires a single run of the permutation test. It is thus way faster and has no grid-related approximation error.
	
	\item An estimation/correction procedure for multivariate confidence intervals that can handle and exploit arbitrary dependence structures. This allows for a much less conservative correction, strengthening the statistical inference and conclusion hereof. 
	
	In fact, a high degree of correlation is typical for multivariate data. Various methods exist for parametric models, where one can focus on single parameters. Contrary, non-parametric multivariate methods typically merely use "positive association" (for which uncorrelated data is the border case) and do not take the degree of correlation into account. 

	In the simulation experiment described in Section 3, we varied the correlation from $0.9$ to $0.99$. The associated coverage and  adjusted confidence level changed accordingly. 
\end{itemize}

\section{Methodology}

Our methodology concerns the construction of confidence intervals, and thus we are interested in formulations of the kind $P_\theta(\theta \in I)$,  where $I$ is a an interval and a measurable function of the outcome $X$. Since we are working in the realm of permutation tests, we will at times condition on some $\sigma$-algebra $\M{D} \subseteq \mathcal{F}(X)$, where $\mathcal{F}(X)$ is the $\sigma$-algebra generated by $X$.

Though we might only be able to determine $P_\theta(\theta \in I | \M{D})$ for a non-trivial $\M{D}$, it holds that
$$
P_\theta(\theta \in I) = \E[P_\theta(\theta \in I | \M{D})]
$$
and thus $P_\theta(\theta \in I | \M{D})$ is an unbiased estimate of $P_\theta(\theta \in I) $. Additionally, $P_\theta(\theta \in I | \M{D}) = \alpha$ for some constant $\alpha$ is in fact a \emph{stronger} statement than $P_\theta(\theta \in I) = \alpha$ and is robust to model misspecifications only related to $\M{D}$. 

For permutation sets, one commonly uses the \emph{conditional reference space} which we here for a random variable $X: (\Omega, \M{F})\pil \R^N$ define as the $\sigma$-algebra
$$
\bar{\mathcal{F}}(X) =  \{X^{-1}(A) : A \in \M{B} \cap \M{S} \}
$$
where $\M{B}$ is the Borel-$\sigma$-algebra, and $\M{S}$ is the $\sigma$-algebra of \emph{symmetric sets}: $$S \in \M{S} \eb [x \in S \Rightarrow s(x) \in S \quad \forall s \in S_N, x \in \R^N]$$
The intuition behind $\bar{\mathcal{F}}(X)$ is that information about $X$ is known only up to permutation, e.g. $X = (1,2)$ and $X = (2,1)$ are indistinguishable.

\medskip
\paragraph{Confidence intervals}
As we in this work (formally) are considering confidence by means of a type I error risk (ie. a $p$-value), we shall implicitly assume the confidence interval as part of an ordered family of intervals. This implicit ordering  is satisfied by the common methods  for constructing confidence intervals.
Furthermore, for multivariate parameters, we wish to consider type I errors for different coordinates separately, that is, $\theta_k \notin I_k$ is \emph{specifically} a rejection of the hypothesis $H_0: \theta^0_k = \theta_k$.  
This motivates the quite heavy definition of Definition \ref{coverage-defi-mv}.
 
\begin{defi}[Confidence interval] \label{confint-defi} Let $X \in \mathcal{X}$ be a random variable generated by the statistical model $\{P_{\theta}, \theta \in \R\}$,
	where  $\theta \in \R$ is an unknown parameter of interest. We remark that the statistical model can depend on other unknowns, but these are considered fixed and thus omitted from the model. 
	
 We define a \emph{confidence interval series} (for $\theta$) as a family of intervals $\{I_\alpha = [a_\alpha, b_\alpha] | 1 > \alpha > 0\}$ with the property that each $I_\alpha$ is a measurable function of $X$ and $\alpha_1 < \alpha_2 \Rightarrow I_{\alpha_2} \subseteq I_{\alpha_1}$.
	
A \emph{confidence interval} at (nominal) level ($1- \alpha$) is an interval $I = [a,b]$ with an implicit understanding that $I = I_\alpha$ for some $I_\alpha$ in a confidence interval series.
\end{defi}
For example in the one-sample normal model, we can express the textbook confidence interval for $\theta$ as a confidence interval series by:
\begin{equation}
\{ I_\alpha =  [\hat{\theta} -  t_{1 - \alpha/2} \cdot \tfrac{s}{\sqrt{N}}, \hat{\theta} +  t_{1 - \alpha/2} \cdot \tfrac{s}{\sqrt{N}} ] ,  \alpha \in (0,1) \} \label{ci-eks1}
\end{equation}
 Here the statistical model is $\{N(\theta, \sigma^2)^{\otimes N} | \theta \in \R \}$ with unknowns $\theta$ and $\sigma^2$.

\begin{defi}[coverage and type I risk, univariate] \label{coverage-defi}
	
Continuing  the setting	of definition \ref{confint-defi}, assume that we observe a confidence interval series $\{I_\alpha\} = \{I_\alpha(x) \}$ corresponding to an observation $x \in \mathcal{X}$, and let $I \in \{I_\alpha\}$ be a confidence interval.

	 We define the \emph{type I risk} (or $p$-value) for $\theta$ conditionally on $\M{D}$ as 
\begin{equation}
\inf_{(\alpha: \ \theta \notin I_\alpha)} P_\theta(\theta \notin J_\alpha | \M{D}) 
\label{risk1}
\end{equation}
where $J_\alpha$ is seen as a random variable, 
and define the {type I risk} of $I$ as $\sup\{$type I risk$(\theta) | \theta \notin I\}$. 
We then define the {coverage} of $I$ as $1 -  \text{type I risk}$.
Due to the ordering property of confidence intervals, the infimum in \eqref{risk1} will be attained in the "limit" of $I_\alpha$s not containing $\theta$. 
\end{defi}
For example, it is easily verified that for a given $\alpha$, $I_\alpha$ from \eqref{ci-eks1} has coverage $1 - \alpha$. However, were we to choose $\M{D} = \mathcal{F}(\hat{\theta}, s)$, ie. the sufficient statistic, then $P_\theta(\theta \notin J_\alpha | \M{D}) \in \{0,1\}$, which is sort of meaningless from an inference perspective. \medskip

For parameters in $\R^K$ we shall consider the type I errors for different coordinates separately, for instance $\theta_k \notin I^k$ is a rejection of the hypothesis $H_0: \theta^0_k = \theta_k$. This leads to the following definition of coverage when having multiple confidence intervals: 

\begin{defi}[coverage and type I risk, multivariate] \label{coverage-defi-mv}
Let $\theta = (\theta_1, \dots, \theta_K) \in \R^K$ be $K$ unknown parameters of interest for a statistical model $\{P_\theta | \theta \in \R^K\}$ that generates $X \in \mathcal{X}$, and assume that to each coordinate of $\theta$ is associated a confidence interval series  $\{I_\alpha^k\}$. 
As for the univariate case, the statistical model can depend on other omitted, but  fixed unknowns.

Assume that we observe coordinate-wise confidence interval series $\{I_\alpha^k\}$ corresponding to an observation $x \in \mathcal{X}$, and let $I = I^1 \prikt{\times} I^K$, $I^k \in \{I^k_\alpha\}$ be a confidence region. 
	
We
define the type I risk (or adjusted $p$-value) for $\theta= (\theta_1, \dots, \theta_K)$ at coordinates $\tilde{K} \subseteq \{1, \dots, K\}$ as 
\begin{equation}
\inf_{(\alpha_1 \prikt{,} \alpha_K: \  \theta_k \notin I^k_{\alpha_k}\;  \forall k \in \tilde{K})} P_\theta \left(\bigcup_{k \in \tilde{K}} (\theta_k \notin J^k_{\alpha_k})  \middle| \M{D} \right) 
\label{risk1-multi}
\end{equation}
and define the (joint)  type I risk of $I$ or $I^1 \prikt{,} I^K$ as
$$
\sup\{\text{type I risk}(\theta) \text{ at coordinates }  \tilde{K} \text{ where } \tilde{K} = \{k: \theta_k \notin I^k \}
| \theta \notin I\}
$$
In other words, the type I risk is the chance of making any type I error under $H_0: \theta_0 = \theta$ when using $I$ for inference. 
\end{defi} 
As an example, the coverage of the usual 95\% confidence interval for a single parameter in the linear normal model is $0.95$, but the joint coverage of 95\% confidence intervals is less than $0.95$. In case of independence between coordinates, the joint coverage of $K$ independent $(1-\alpha)$ confidence intervals is $(1-\alpha)^K$.

\subsection{Confidence interval for a single parameter} \label{sect-single-par}

For notation, let  $S_N$ denote the symmetric group of order $N$. We shall identify a permutation $s \in S_N$ with its corresponding permutation function $\R^N \pil \R^N$. We shall use $e$ to refer to the identity permutation.

\paragraph{Statistical model}
Assume $N$ observations $X_1, \dots, X_N \in \R$. Here $X_{i} = \phi_i(\theta) + \epsilon_i$ for an unknown parameter of interest $\theta \in \R$  and an a priori known 'covariate function' $\phi_i: \R \pil \R$. 
For example,  $\phi_i(\theta) = \theta x_i$ for a simple linear regression on the covariate $x = (x_1, \dots, x_N)$.
We assume the residuals $\epsilon_1, \dots , \epsilon_N \in \R^k$ to have the \emph{exchangeability} condition. That is, 
$$
(\epsilon_1, \dots, \epsilon_N) \stackrel{D}{=} s(\epsilon_1, \dots, \epsilon_N) , \quad s \in S_N
$$
[but otherwise we do not put any additional restrictions on $\epsilon$. ]
The common sufficient criterion for exchangeability is that $\epsilon_1, \dots, \epsilon_N$ are i.i.d. We refer to \cite{pesarin} for a discussion.

\paragraph{Test statistics}
We shall assume that we are given a test statistic $t: \R^N \pil \R$. It follows from the properties below that $t$ is a one-tailed statistic for which large values of $t$ are considered extreme. 

Let $\hat{\theta}$  be the $\theta$ which minimises $\theta \mapsto t(X_1 - \phi_1({\theta}), \dots, X_N -  \phi_N({\theta}))$. We will interpret and refer to $\hat{\theta}$ as the \emph{estimate} of $\theta$.

We shall assume that the following properties hold true with probability one for all $s \in S_N$, except for a "negligible" set of permutations (discussed below):

\begin{enumerate}
	\item Minimality of the unpermuted data in $\hat{\theta}$:
	$$
	t(X_1 - \phi_1(\hat{\theta}), \dots, X_N -  \phi_N(\hat{\theta})) <
	t \circ s(X_1 - \phi_1(\hat{\theta}), \dots, X_N -  \phi_N(\hat{\theta})) 
	$$	
	\item Monotonicity:
	$$
	\theta \mapsto t (X_1 - \phi_1(\theta), \dots, X_N -  \phi_N(\theta)) - 
	t \circ s(X_1 - \phi_1(\theta), \dots, X_N -  \phi_N(\theta)) 
	$$	
	is strictly decreasing for  $\theta < \hat{\theta}$ and strictly increasing for $\theta > \hat{\theta}$.
	\item Eventual "significance":
	$$
	\begin{aligned}
	\liminf_{\theta \pil - \infty} t (X_1 - \phi_1(\theta), \dots, X_N -  \phi_N(\theta)) - 
	t \circ s(X_1 - \phi_1(\theta), \dots, X_N -  \phi_N(\theta))  > 0 \\
	\liminf_{\theta \pil \infty} t (X_1 - \phi_1(\theta), \dots, X_N -  \phi_N(\theta)) - 
	t \circ s(X_1 - \phi_1(\theta), \dots, X_N -  \phi_N(\theta))  > 0
	\end{aligned}
	$$
\end{enumerate}
Since the above properties are not valid for all $s \in S_N$ (e.g. by selecting $s = e$), we have to consider a "negligible" set ${V} \subset S_N$, for which the above property does not hold. The negligibility criterion is to be interpreted as $\frac{\# V}{\# S_N}$ being small, preferably much smaller than the significance level $\alpha$.

\begin{figure}
	\centering
	\includegraphics[trim = 0 30 0 0, clip]{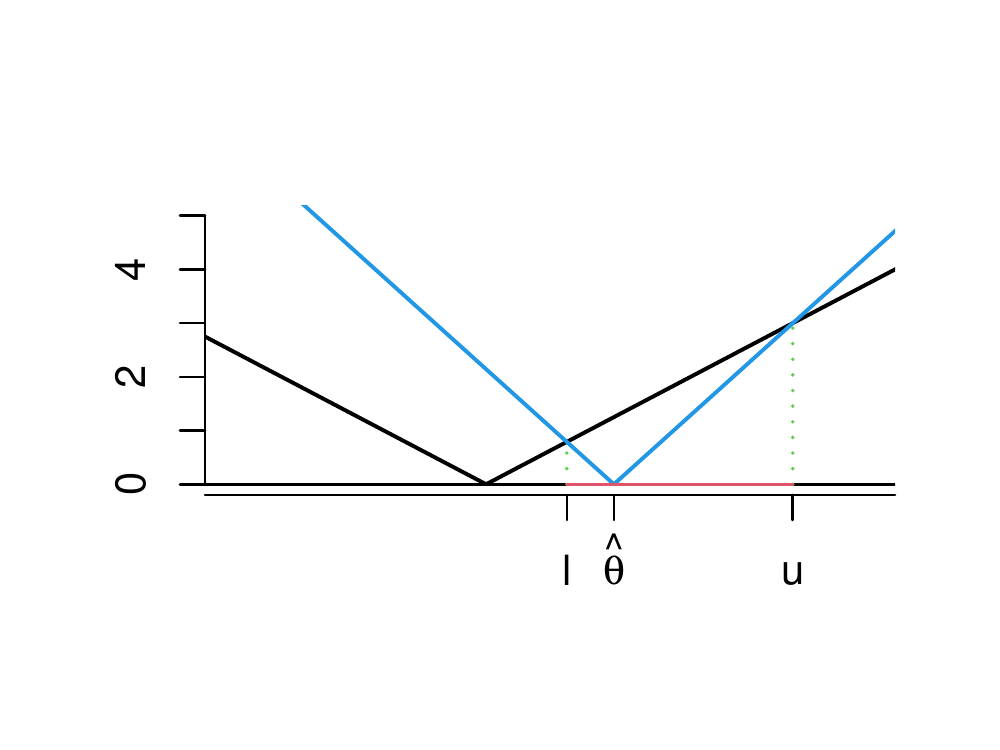}
	\caption{Illustration of the test statistic as a function of $\theta$.
		Blue is the test statistic for the unpermuted data, black is  the test statistic for a non-negligible permutation, and $l$ and $u$ are the interval limits from the algorithm.}
	\label{fig-test-stat}
\end{figure}

\paragraph{Pointwise confidence intervals}

Let $s \in S_N$ be a non-negligible permutation. 
From the properties (1) - (3) above, it holds that there exists an interval $[l,u] \subset \R$ such that:
$$
t (X_1 - \phi_1(\theta), \dots, X_N -  \phi_N(\theta)) >
t \circ s (X_1 - \phi_1(\theta), \dots, X_N -  \phi_N(\theta)) 
$$	
iff $\theta \notin (l,u)$. Furthermore, $\hat{\theta} \in (l,u)$. See Figure \ref{fig-test-stat} for an illustration. \medskip

We now define a confidence interval $[L, U]$ of nominal level $1- \alpha$. 
 Our algorithm consists of two steps:
\begin{enumerate}
	\item Let $s_1, \dots , s_M \in S_N$ be random permutations. For $m = 1,  \dots, M$, define $l_m$ and $u_m$ as the interval limits above, and set $l_m = - \infty, u_m = \infty$ when $s_m$ is negligible. 
	
	\item[2a.] Define $L$ as the $\alpha$ quantile of $(l_1, \dots, l_M)$.
	
	\item[2b.] Define $U$ as the $(1 - \alpha)$ quantile of $(u_1, \dots, u_M)$. %
\end{enumerate}
This construction satisfies $[L_1, U_1] \subseteq [L_2, U_2] \eb \alpha_2 < \alpha_1$ (when using the same set of permutations), and thus satisfies the criteria of Definition 1. \medskip

For the proof that the procedure works, we need the following lemma, which connects the confidence intervals with quantiles of a permutation test. We use the following definition of quantile in the lemma: $Q_{\gamma}(a_1, \dots, a_M) = a_{(\lceil M\gamma \rceil)}$.

\begin{lemma} \label{lemma-et} Let $[L,U]$ be a  $(1 - \gamma)$ confidence interval constructed using the algorithm above. Define
	$$T_\text{obs} = t(X_1 - \phi_1(\theta), \dots, X_N -  \phi_N(\theta)) $$
	and
	$$T_m = t \circ s_m(X_1 - \phi_1(\theta), \dots, X_N -  \phi_N(\theta)), \quad  m = 1, \dots, M  $$
	
	Let $\gamma \in (0,1)$. Then $\theta \notin I_\gamma$ iff a fraction at most $\gamma$ of 
	$\{T_m\}_{m =1}^M$
	are larger than $T_\text{obs}$, ie. $M^{-1}\#\{m: T_m > T_\text{obs} \} \leq \gamma $. 
	In other words, $\theta \notin I_\gamma$ iff the associated hypothesis $H_0: \theta_0 = \theta$ is rejected.
	
\end{lemma}

\begin{proof}
	Set $ I_\gamma = [L_\gamma, U_\gamma]$. 

	Assume  $\theta > \hat{\theta}$ which means that we have the relation $\theta \notin I_\gamma \eb \theta > U_\gamma$. The case of $\theta < \hat{\theta}$ is analogous. 
	Observe that
	$$
	T_\text{obs} > T_m \eb \theta > u_m
	$$
	Thus when 
	$$
	T_\text{obs} > T_{(j)} \eb \theta > u_{(j)} \eb \theta > Q_{j/M}(u_1, \dots, u_M)
	$$
	Now using $U_\gamma = Q_{1-\gamma}(u_1, \dots, u_M) = u_{(\lceil M(1-\gamma) \rceil)}$. 
	$$
	\theta > U_\gamma \eb \theta > u_{(\lceil M(1-\gamma) \rceil)} \eb T_\text{obs} > T_{(\lceil M(1-\gamma) \rceil)} 
	$$
	At most $\lfloor M\gamma \rfloor$ of $(T_1, \prikt{,} T_M)$ are larger than $T_{(\lceil M(1-\gamma) \rceil)}$, showing the claim. 
	
\end{proof}

\begin{prop} \label{coverage-prop} 
	Let $[L,U]$ be a  $(1 - \alpha)$ confidence interval constructed using the algorithm above. The $[L,U]$ has a coverage of at least $(1 - \alpha)$ conditionally on $\bar{\mathcal{F}} = \bar{\mathcal{F}}(X_1, \prikt{,} X_N)$, when $\{s_1, \dots, s_M\} = S_N$.

	Additionally, 
\begin{equation}
	P_{\theta}(\theta  \notin [L,U] ) \leq \alpha 
	\label{prop1-0.95}
\end{equation}
where $[L,U]$ is viewed as a random variable. 

With the proposed method of using random permutations for $\{s_1, \dots, s_M\}$, \eqref{prop1-0.95} is an approximation which becomes increasingly correct as $M \pil \infty$. See also Section \ref{sect-uncertain}. 
\end{prop}
The  usage of random sampling in Algorithm 1 is referred in \cite{pesarin} as the \emph{Conditional Monte Carlo} method and is a practical need in permutation tests due to the infeasibility of evaluating all $N!$ permutations in $S_N$ for all but the smallest $N$.

\begin{proof} Use $X = (X_1, \prikt, X_N)$.
Let $\theta \notin [L,U]$. We must show that the type I risk for $\theta$ is less than $\alpha$. 

Define 
$$T_\text{obs} = t(X_1 - \phi_1(\theta), \dots, X_N -  \phi_N(\theta)) $$
and
$$T_m = t \circ s_m(X_1 - \phi_1(\theta), \dots, X_N -  \phi_N(\theta)), \quad  m = 1, \dots, M  $$

Let $\gamma \in (0,1)$. By Lemma \ref{lemma-et}, $\theta \notin I_\gamma$ iff a fraction at most $\gamma$ of 
$\{T_m\}_{m =1}^M$
are larger than $T_\text{obs}$, ie. $M^{-1}\#\{m: T_m > T_\text{obs} \} \leq \gamma $. We shall therefore evaluate the probability 
\begin{equation}
P(M^{-1}\{m: T_m > T_\text{obs} \} \leq \gamma |\bar{\mathcal{F}}) \label{tobs-tm-prob}
\end{equation}
for $\gamma \leq \alpha$. 
Now assume $H_0: \theta_0 = \theta$ is true. Then the distribution of 
$$\mathcal{T}:= t \circ s(X_1 - \phi_1(\theta), \dots, X_N -  \phi_N(\theta)) $$
is unchanged by  $s \in S_N$, also conditionally on $\bar{\mathcal{F}}$. So if let $S$ be a random sample from $S_N$, $S \uaf X$, we get
\begin{equation}
P(M^{-1}\#\{m: \tilde{T}_m > \tilde{T}\} \leq \gamma | \bar{\mathcal{F}}) = 
P(M^{-1}\#\{m: {T}_m > \tilde{T}\} \leq \gamma | \bar{\mathcal{F}}) \label{ttilde-tm-prob}
\end{equation}
where
$$\tilde{T} = t \circ S(X_1 - \phi_1(\theta), \dots, X_N -  \phi_N(\theta)) $$
and
$$\tilde{T}_m = t \circ s_m \circ S(X_1 - \phi_1(\theta), \dots, X_N -  \phi_N(\theta)), \quad  m = 1, \dots, M  $$
The equality in \eqref{ttilde-tm-prob} follows from the fact that the $\tilde{T}_m$s are a bijection of the ${T}_m$s. We will now condition on $X$, under which the $T_m$s are no longer stochastic:
\begin{equation}
P(M^{-1}\#\{m: {T}_m > \tilde{T}\} \leq \gamma| \bar{\mathcal{F}} ) =
\E[P(M^{-1}\#\{m: {T}_m > \tilde{T}\} \leq \gamma | \mathcal{F}(X) )| \bar{\mathcal{F}}]
\end{equation}
Conditionally on $X$, $\tilde{T}$ randomly attains one of the $T_m$s, counted with multiplicity. Therefore
\begin{equation}
P(M^{-1}\#\{m: {T}_m > \tilde{T}\} \leq \gamma | \mathcal{F}(X)) \leq \gamma
\end{equation}
and hence
\begin{equation}
P(\theta \notin J_\gamma | \bar{\mathcal{F}} ) = \E[P(M^{-1}\#\{m: {T}_m > \tilde{T}\} \leq \gamma | \mathcal{F}(X) ) | \bar{\mathcal{F}}] \leq \gamma
\end{equation}
where $J_\gamma$ is $I_\gamma$ seen as a random variable. Since $\gamma$ was assumed smaller than $\alpha$, the result follows.

The result \eqref{prop1-0.95} follows as an immediate consequence. 
\end{proof}
Below follows two  examples of statistical models; the two-sample case can be seen as a special case of the linear regression. 

\begin{eks}[Two-sample test] \label{eks-two-sample}
	Assume $Y_1, \dots Y_{n1}, Z_1 \dots, Z_{n2}$ are two samples with different means and i.i.d. errors, commonly referred to as the (unpaired) two-sample setup.
	
	In detail,
	$$
	Y_i = \mu_Y + \epsilon_{1i}, \quad Z_j = \mu_Z + \epsilon_{2j}, \quad i = 1, \dots, n_1, j = 1, \dots, n_2
	$$
	where all $\epsilon_{\cdot \cdot} \sim D$ i.i.d. for an unknown distribution $D$. 
	We wish to infer a confidence interval for difference in means, $\theta = \mu_Y - \mu_Z$.
	
	We can now use Algorithm 1 with covariate function $\phi$ and test statistic $t$ given by
	$$
	\phi_i(\theta) = \begin{cases} \theta & i = 1, \dots, n_1 \\
	0 & i = n_1 + 1, \dots, n_1 + n_2
	\end{cases}, \quad 
	t(X) = |\bar{X}_Y - \bar{X}_Z| 
	$$		
	where $\bar{X}_Y $ is the average of the first $n_1$ values and $\bar{X}_Z $ is the average of the remaining $n_2$ values.
	Then $t$ satisfies the properties (1)-(3) above, and the estimate of ${\theta}$ is given by $\hat{\theta} = \bar{Y} - \bar{Z}$.
	
Assume $n_1 > n_2$. The set of negligible permutations consists of those permutations that map $\{1, \dots, n_1\}$ to $\{1, \dots, n_1\}$. There are $n_1!n_2!$ such permutations; thus the fraction of negligible permutations is 
	$$
	\frac{\# V}{\# S_{n_1 + n_2}} = \frac{n_1!n_2!}{(n_1 + n_2)!} = 1 / \binom{n_1}{n_1 + n_2}
	$$
	which is small and goes rapidly towards zero for increasing sample sizes. 

\end{eks}

\begin{eks}[Linear regression] \label{eks-lin-reg-test}
Here we consider the confidence interval for $\beta$ in the linear regression model, $y = \alpha + \beta x + \epsilon$.
	In detail, the statistical model is
$$
Y_i = \alpha + \beta \cdot x_i + \epsilon_i, \quad i = 1, \dots, N
$$	
where $\epsilon_i \sim D$ i.i.d. for an unknown distribution $D$, and $x_1, \dots,  x_N$ are regressor values. 
	
	We can now use Algorithm 1 with covariate function $\phi$ and test statistic $t$ given by
$$
\phi_i(\theta) = \theta x_i , \quad
t(\epsilon_1, \dots, \epsilon_N) = \left| \sum_{i=1}^{N} (x_i - \bar{x})(\epsilon_i - \bar{\epsilon}) \right|
$$	
Then $t$ satisfies the properties (1)-(3) above, and the estimate of ${\beta}$ is given by the usual least squares estimator; ie $\hat{\beta} = \frac{\sum_{i=1}^N (x_i - \bar{x})(y_i - \bar{y})}{\sum_{i=1}^N (x_i - \bar{x})^2}$.
	
\paragraph{Negligible permutations} The negligible permutations for the linear regression are exactly those permutations for which $s(x) - \bar{x} = \pm (x - \bar{x})$ (see the appendix).
This in general depends on the experimental setup; ie. the $x$ values. As for the two-sample test, the fraction of negligible permutations decreases rapidly towards zero for increasing sample sizes.  		
\end{eks}

\subsection{Simultaneous confidence intervals for multiple testing} \label{sect-multi-par}
In this section we consider the scenario of confidence intervals under multiple testing. We will assume $K$ parameters $\theta_1 , \dots , \theta_K \in \R$ and $N$ observations $X_1 , \dots , X_N \in \R^K$. We impose the model of Section \ref{sect-single-par} on each coordinate, ie. $X_{ik} = \phi_{ik}(\theta_k) + \epsilon_{ik}$.
There can be arbitrary dependence between coordinates, but $\epsilon_1 , \dots , \epsilon_N$ must be jointly exchangeable:
$$
(\epsilon_1 , \dots , \epsilon_N) \stackrel{D}{=} s(\epsilon_1, \dots, \epsilon_N), \quad s \in S_N
$$ 

We assume that we are given a test statistic $t_k$ for each coordinate $k = 1, \dots, K$, such that $t_k$ satisfies the conditions described in Section \ref{sect-single-par}. Applying Algorithm 1 jointly on the coordinates (ie. using the same (random) permutations $s_1, \dots, s_M$) then produces $(1-\alpha)$ confidence intervals $(L_1, U_1), \dots, (L_K, U_K)$.

We now consider the two following aspects: 
\begin{enumerate}
	\item What is joint coverage level of $(L_1, U_1), \dots, (L_K, U_K)$?
	\item How do we adjust  $(L_1, U_1), \dots, (L_K, U_K)$ such that the joint coverage is $(1 - \alpha)$?
\end{enumerate}

\paragraph{Computing the joint coverage level for a given $\alpha$}

Though the confidence intervals $[L_1, U_1], \dots, [L_K, U_K]$ each have $(1-\alpha)$ coverage, the joint coverage is less than  $(1-\alpha)$. 

Let $C = \{L_1, U_1\} \times \dots \times \{L_K, U_K\}$ denote the corners of $B$, and let $\hat{\theta} = (\hat{\theta}_1, \dots, \hat{\theta}_K)$ denote the joint estimate of ${\theta}_0$.

We calculate the joint coverage $\alpha_\text{multiple}$ according to the following algorithm: 
\begin{enumerate}
	\item For $i = 1, \dots, M$ and $k = 1, \dots, K$, define $l_{ik}$ and $u_{ik}$ as in Algorithm 1.
	
	\item For each $c \in C$, we calculate the number of instances $R_c$ for which
	$$
	\text{at least one of } 
	\begin{cases}
	l_{nk} \in [c_k , \hat{\theta}_k] & c_k = L_k \\
	u_{nk} \in [\hat{\theta}_k, c_k] & c_k = U_k
	\end{cases}	
	\text{ is false}, \quad n = 1, \dots, N.
	$$
	
	\item Then we set $\alpha_\text{multiple} = \max_{c \in C} R_c / M$.
\end{enumerate}

\begin{prop}
	The joint coverage of $(L_1, U_1), \dots, (L_K, U_K)$ conditional on $\bar{\mathcal{F}} = \bar{\mathcal{F}}(X_1 \prikt{,} X_N)$ is at least $1 - \alpha_\text{multiple}$, when $\{s_1, \dots, s_M\} = S_N$. 
\end{prop}

\begin{proof}

Set $B = [L_1, U_1] \times \dots \times [L_K, U_K]$ and let $\theta = (\theta_1, \dots, \theta_K)  \in \R^k$. We define 
\begin{align*}
T_\text{obs}^k &= t_k (X_{1k} - \phi_{1k}(\theta_k), \dots, X_{Nk} -  \phi_{Nk}(\theta_k)), \quad k = 1, \dots, K \\
T_\text{obs} &= (T_\text{obs}^1, \dots, T_\text{obs}^K) \\
\intertext{and}
T_m^k &= t_k \circ s_m(X_{1k} - \phi_{1k}(\theta_k), \dots, X_{Nk} -  \phi_{Nk}(\theta_k)), \quad k = 1, \dots, K, m = 1, \dots M \\
T_m &= (T_m^1, \dots, T_m^K), \quad m = 1 \prikt{,} M
\end{align*}

If $\theta \notin B$, we have the risk of making one or more type I errors. We must verify that this risk is less than $\alpha_\text{multiple}$.

So assume $\theta_k \notin [L_k, U_k]$ for a non-empty subset of $\{1, \dots, K\}$. Without loss of generalisation we can assume $\theta \notin [L_k, U_k]$ for $k = 1, \dots, \tilde{k}$.

Consider the number $R$ for which 
\begin{equation}
R = \stackrel[j = 1 \prikt{,} M]{}{\#} (\stackrel[k = 1, \dots, \tilde{k}]{}{\text{at least one of }} \: T_j^k > Q_{1- \alpha}(T^k_1 \prikt{,} T^k_M )) \label{eq-anyofk}
\end{equation}
By construction of the confidence interval $I^k$ and the definition of $\alpha_\text{multiple}$, it holds that $R < R_c$.

We shall evaluate the probability
\begin{equation*} P_\theta \left(\bigcup_{k \in \tilde{K}} (\theta_k \notin J^k_{\gamma_k}) \:|
\bar{\mathcal{F}}
\right) 
\end{equation*}
where $\gamma_k \leq \alpha$ for all $k \in \tilde{K}$. 

Following Lemma 1, 
$$
(\theta_k \notin J^k_{\gamma_k}) \eb M^{-1} \# \{m:  T_m^k > T_\text{obs}^k\} \leq \gamma_k
$$
Assume the true value of $\theta_0$ is $\theta_0 = \theta$. Then the distribution of 
$$\mathcal{T}_k = t_k \circ s(X_1 - \phi_1(\theta_k), \dots, X_N -  \phi_N(\theta_k)), \quad k = 1, \dots, K
$$
is unchanged by  $s \in S_N$. So if let $S$ be a random sample from $S_N$, $S \uaf X$, we get

\begin{multline} P_\theta \left( \bigcup_{k \in \tilde{K}} (  M^{-1} \# \{m:  T_m^k > T_\text{obs}^k\} \leq \gamma_k ) | \bar{\mathcal{F}} \right) =\\ 
P_\theta \left( \bigcup_{k \in \tilde{K}} (  M^{-1} \# \{m:  \tilde{T}_m^k > \tilde{T}^k\} \leq \gamma_k ) | \bar{\mathcal{F}} \right) = \\
P_\theta \left( \bigcup_{k \in \tilde{K}} (  M^{-1} \# \{m:  {T}_m^k > \tilde{T}^k\} \leq \gamma_k ) | \bar{\mathcal{F}} \right) \label{p-gamma-k}
\end{multline}
where
\begin{align*}
\tilde{T}^k &= t_k \circ S(X_{1k} - \phi_{1k}(\theta_k), \dots, X_{Nk} -  \phi_{Nk}(\theta_k)), \quad k = 1, \dots, K \\
\tilde{T} &= (\tilde{T}^1, \dots, \tilde{T}^K) \\
\intertext{and}
\tilde{T}_m^k &= t_k \circ s_m \circ S (X_{1k} - \phi_{1k}(\theta_k), \dots, X_{Nk} -  \phi_{Nk}(\theta_k)), \quad k = 1, \dots, K, m = 1, \dots M \\
\tilde{T}_m &= (\tilde{T}_m^1 \prikt{,} \tilde{T}_m^K), \quad m = 1 \prikt{,} M
\end{align*}

The equality in \eqref{p-gamma-k} follows from the fact that the $\tilde{T}_m$s are a bijection of the ${T}_m$s. We will now condition on $X$, under which the $T_m$ are no longer stochastic:
\begin{multline} P\left(\bigcup_{k \in \tilde{K}} (  M^{-1} \# \{m:  {T}_m^k > \tilde{T}^k\} \leq \gamma_k ) \middle| \bar{\mathcal{F}}(X_1 \prikt{,} X_N) \right) =  \\
\E[P\left(\bigcup_{k \in \tilde{K}} (  M^{-1} \# \{m:  {T}_m^k > \tilde{T}^k\} \leq \gamma_k ) | \mathcal{F}(X) \right) | \bar{\mathcal{F}}(X_1 \prikt{,} X_N) ] \label{ep-multi}
\end{multline} 
Conditionally on $X$, $\tilde{T}$ randomly attains one of the $T_m$s, counted with multiplicity, and $R$ defined in \eqref{eq-anyofk} is non-random. 
Therefore 
\begin{equation}
P\left(\bigcup_{k \in \tilde{K}} (  M^{-1} \# \{m:  {T}_m^k > \tilde{T}^k\} \leq \alpha ) \middle| \mathcal{F}(X)\right) = R/M \label{R-formel}
\end{equation}
As we in \eqref{R-formel} are considering a larger set compared to \eqref{ep-multi}, it holds that
$$
P\left(\bigcup_{k \in \tilde{K}} (  M^{-1} \# \{m:  {T}_m^k > \tilde{T}^k\} \leq \gamma_k ) \middle| \mathcal{F}(X) \right) \leq R/M \leq R_c/M = \alpha_\text{multiple}
$$
which ends the proof. 
\end{proof}
Note that we do not have an unconditional probability statement similar to equation \eqref{prop1-0.95} as this would require us to know the full copula of $(T_1, \dots, T_M)$ for every value of $\theta \in \R^K$. 

\paragraph{Adjusting the confidence level}
Complementing the multi-confidence level, we can adjust confidence intervals to a level $\alpha_\text{multiple}$, such that the multi-confidence level is $\alpha$. 

The procedure is straightforward: 
\begin{itemize}
	\item For a given $\alpha^*$, calculate $\alpha^*_\text{multiple}$
	\item Adjust $\alpha^*$ until $\alpha^*_\text{multiple} = \alpha$ or $ |\alpha^*_\text{multiple} - \alpha|$ is less than a given threshold. 
\end{itemize}

\subsection{Computational issues}  Let $M$ denote the number of permutations and $K$ the number of parameters.
Then the confidence interval for a single parameter has a computational cost which in principle is  $O(M \log M)$. The $\log M$ factor is due to the sorting of $l$ and $u$ values. Since sorting usually is very fast, the "practical" computational cost is $O(M)$, similar to  usual permutation tests.

However, the multiple testing procedure has computational cost $O(2^K)$ (for a fixed $M$). This is due to every corner in $[L_1, U_1] \prikt{\times} [L_K , U_K]$ being evaluated. This imposes a practical constraint on the size of $K$, though for at least $K = 15$ this should not be an issue.

\subsection{Uncertainties in confidence interval calculation} \label{sect-uncertain}
Due to the fact that our method involves random permutations, there will be some uncertainty in the confidence interval(s), even for a fixed realisation of data. 
 This property is a well-known feature of permutation tests, where this uncertainty decreases by increasing the number of permutations $M$.
We suggest/advise to use bootstrapping of the quantile vectors $l$ and $u$ to assess the effect of the random sampling from $S_N$. 

\section{Simulation \& application}

\subsection{Application: Monthly means of Canadian weather data}
In this section we applied the methodology to the well-known "Canadian weather" data set of functional data analysis \cite{ramsay}. 
We considered monthly means of two regions, \emph{Atlantic} and \emph{Continental}, consisting of 15 and 9 observations in $\R^{12}$, respectively. Data are illustrated in Figure \ref{fig-cdata}.

\begin{figure}[!htb]
	\includegraphics[width=0.5\textwidth]{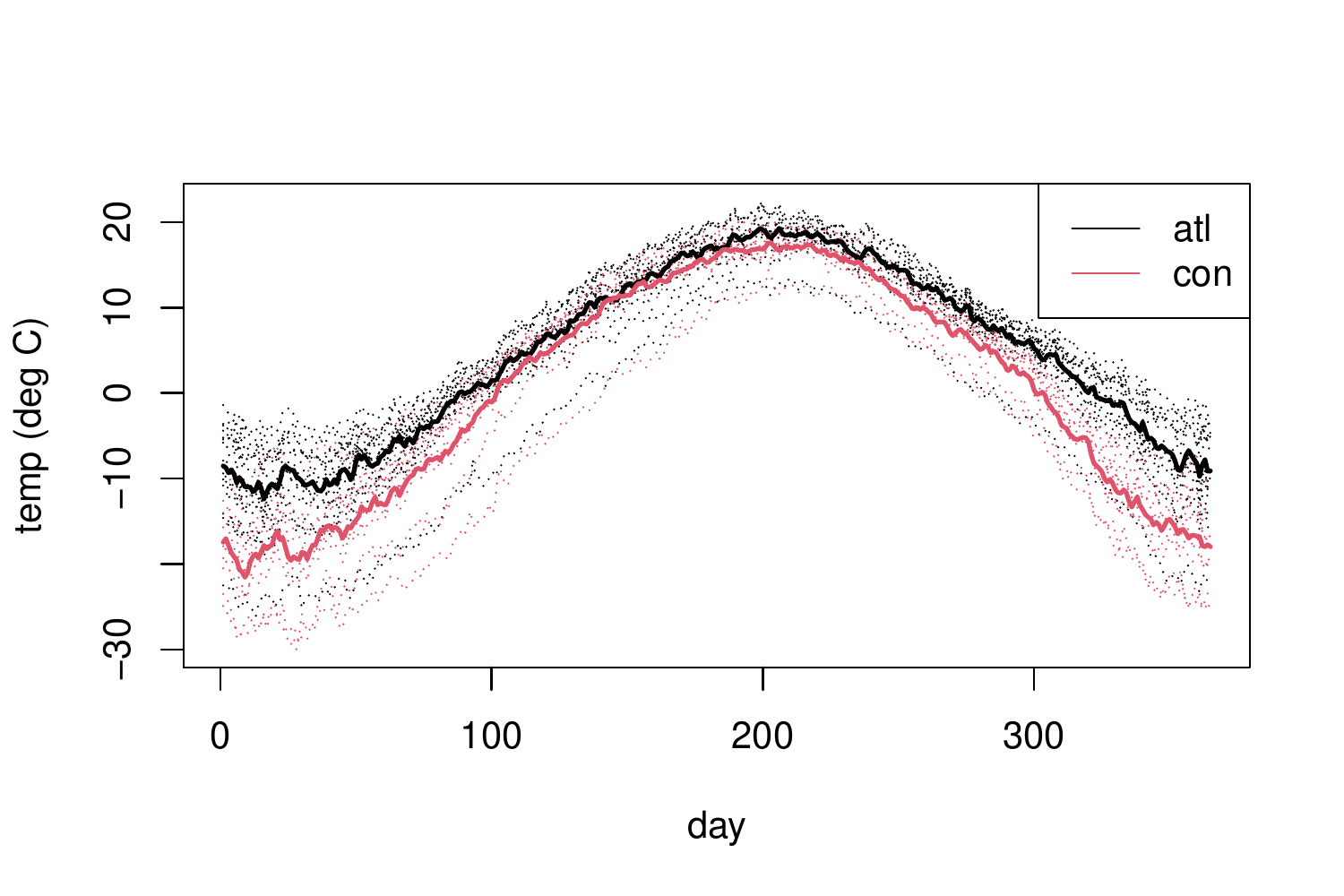}
	\includegraphics[width=0.5\textwidth]{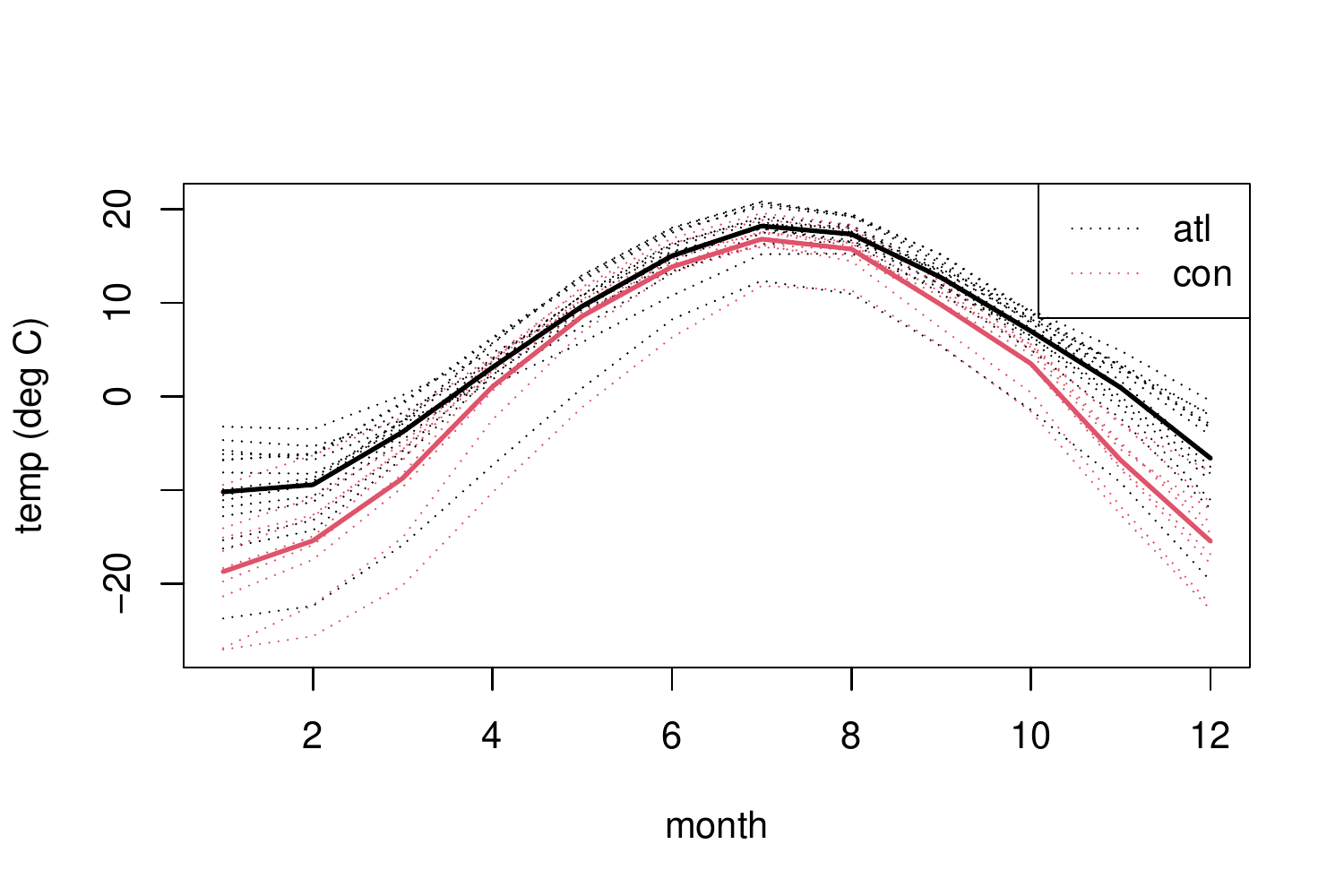}
	\caption{Temperature profiles of 24 Canadian weather stations. Left: Daily averages. Right: Monthly averages. Bold lines indicate group means.} \label{fig-cdata}
\end{figure}

Our parameter of interest is the difference in means,
$$
\theta_i = \mu^i_\text{atlantic} - \mu^i_\text{continental}, \quad i = 1, \dots, 12
$$
where $i$ corresponds to the $i$'th month of the year. 

There is a clear correlation in data as well as heteroscedastic variation, which make parametric methods less applicable. We applied the presented methodology using the two-sample test of Example \ref{eks-two-sample}. We used $M = 10000$ permutations. 

\paragraph{Results}
\begin{figure}[!htb]
	\centering
	\includegraphics[width=0.65\textwidth]{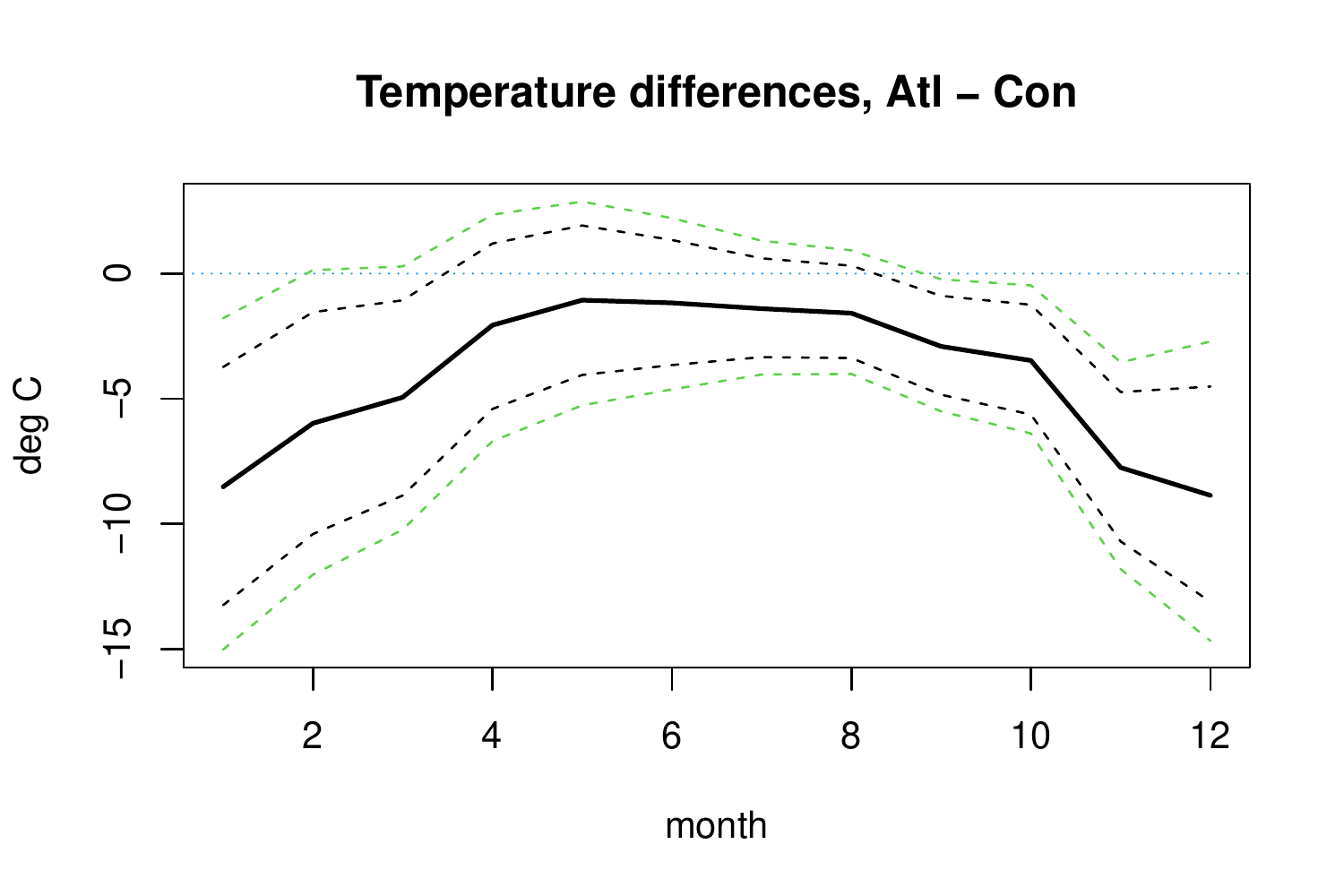}
	\caption{Mean monthly differences and confidence bands using unadjusted (black) and adjusted (green) confidence intervals. Blue dotted line indicates zero (ie. no difference between the two groups).}
\end{figure}

The coverage for the unadjusted 95\% confidence intervals was found to be $79\%$. For comparison, the coverages under the assumptions of the Sidak and Bonferroni procedures would have been $54\%$ and $40\%$, respectively. 
The adjusted confidence region (adjusted so that the coverage is $95\%$) had marginal coverage $99.1\%$, ie. $\alpha^* = 0.009$. 

\subsection{Simulation: Linear regression with strongly correlated outcomes}

We perform a small simulation experiment using a multivariate linear regression with correlated errors. 
Our regressor values $ x = x_1, \dots, x_{20}$ are generated uniformly from $(-1, 1)$; these are fixed for the entirety of the simulation. 

The statistical model is 
$$ Y_i = \alpha + \beta x_i + \epsilon_i, \quad i = 1, \dots 20 $$
with unknown $\alpha, \beta \in \R^8$.

We generate data according to: 
$$
\begin{gathered}
\alpha_1, \dots, \alpha_8 = 0 \\
\beta_1, \dots, \beta_8 = 1 \\
\epsilon_i \sim N(0, D), \quad
[D]_{kl} = \begin{cases}
1 & k = l \\
\rho & k \neq l
\end{cases}
\end{gathered} $$
using $\rho \in \{0.90, 0.95, 0.99\}$.

We inferred confidence intervals for $\beta_1, \dots, \beta_8$, by applying the presented methodology using the test of Example \ref{eks-lin-reg-test}. We used $M = 1000$ permutations for each simulation run, and used 100 simulation runs for each value of $\rho$. We used a threshold of $1/640 \approx 0.0016$ in the calculation of $\alpha^*$. 

\paragraph{Results}
Estimates of $\beta_1, \dots, \beta_8$ are given by the ordinary least squares estimates, and thus their distributions follow the classical theory, ie. $\hat{\beta}_k \sim N(\beta_k, 1/{\sum_{i=1}^N (x_i - \bar{x})^2})$.

Our focus is on the joint coverage of the confidence intervals. We report the mean and inter-quartile range (IQR) of the coverage $\alpha_\text{multiple}$  at $\alpha = 0.05$ (ie. 95\% confidence intervals) and the adjusted confidence level $\alpha^*$ for $\alpha = 0.05$.
\begin{table}[!htb]
	\centering
\begin{tabular}{c|c|c|c|c}
	$\rho$ & mean $\alpha_\text{multiple}$ & IQR $\alpha_\text{multiple}$ & mean $ \alpha^* $ & IQR   $ \alpha^* $ \\ \hline
	0.90 & 0.174 & 0.024 & 0.011 & 0.002 \\
	0.95 & 0.144 & 0.018 & 0.014 & 0.002 \\
	0.99 & 0.114 & 0.009 & 0.018 & 0.003
\end{tabular}
\caption{Coverage and adjusted confidence levels for the simulation} \label{tabel-sim-1}
\end{table}

Results are displayed in Table \ref{tabel-sim-1}. As expected, $\alpha_\text{multiple}$ decreases with increased correlation, and $\alpha^*$ increases correspondingly.

\section{Discussion}
In this paper we have demonstrated a new method for constructing confidence intervals.
We have presented this method in a fairly restricted setting in terms of modelling (the presented examples are linear regression and two-sample comparison), but as permutation tests (including rank tests) have a broader scope, we have strong reason to believe that our methodology extends to these cases as well. 
Secondly, we devised a multiple testing correction procedure, that can handle arbitrary dependencies in the test statistics. We would like to stress the easy implementation and relative speed of the procedure. 

Some readers might argue against using terms "confidence region" and "confidence level" for the multivariate procedure described in Section \ref{sect-multi-par}. However, according to the frequentist interpretation of statistics, the true value $\theta$ is either within or outside the confidence interval, and the probability statement $P(\theta \in \text{CI})$ is understood as the long-term frequency when repeating the experiment. 
Contrary, single realisations of confidence intervals are better understood in terms of controlling the type I risk, and our method falls within this category. Additionally, the univariate method of Section \ref{sect-single-par} also carries the "repeated experiments" interpretation.

Our paper was inspired by the challenge of finding confidence bands for high-dimensional data including functional data. Due to the factor of $2^K$ corners when calculating $\alpha_\text{multiple}$ we have not been able to reach large $K$. We hope that future research can solve this issue and devise a non-parametric method that scales easily to any dimension. 

\section{Software (R package)} An implementation of the procedure is available from GitHub as an R package \url{https://github.com/naolsen/ciperm}.

\section*{Acknowledgements}
I am grateful to Professor Bo Friis Nielsen, Technical University of Denmark, for inputs and comments to the manuscript.

\appendix
\section{Negligible permutations for linear regression}
Using the notation of Example \ref{eks-lin-reg-test}, we show that a permutation $s \in S_N$ is negligible iff $s(x) - \bar{x} = \pm (x - \bar{x})$. 
Let $s \in S_N$ be given, and define $f, g: \R \pil \R_+$:
\begin{align*}
f(\theta) &= t(Y_1 - \phi_1 (\theta), \dots, Y_N - \phi_N (\theta)) = \left|\sum_{i=1}^N (x_i - \bar{x}) (Y_i - \bar{Y} -  \theta (x_i - \bar{x})) \right| \\
g(\theta) &= t \circ s (Y_1 - \phi_1 (\theta), \dots, Y_N - \phi_N (\theta)) = \left| \sum_{i=1}^N (x_i - \bar{x}) (Y_{s(i)} - \bar{Y} - \theta (x_{s(i)} - x)) \right|
\end{align*}
We must verify if properties (2) and (3) of the test statistic are satisfied for $s$.
Since both $f$ and $g$ are linear functions, it suffices to consider their derivatives; ie. $s$ is non-negligible iff $|f'(\theta)| > |g'(\theta)|$. We have
\begin{align*}
|f'(\theta)| &= \left|\sum_{i=1}^N  (x_i - \bar{x}) (x_i - \bar{x}) \right| =
\sum_{i=1}^N  (x_i - \bar{x})^2\\
|g'(\theta)| &= \left|\sum_{i=1}^N  (x_i - \bar{x})  (x_{s(i)} - \bar{x}) \right| 
\end{align*}
If we let $\brak{\cdot, \cdot}$ denote the standard inner product on $\R^N$, then it follows by  the Cauchy-Schwartz inequality that 
\begin{equation}
|g'(\theta)| = | \brak{  s(x) - \bar{x}, x-\bar{x}} |  \leq ||x- \bar{x}|| \cdot ||s(x) - \bar{x} || = ||x- \bar{x}||^2 = |f'(\theta)|
\label{cs-ligning}
\end{equation}
Here we have used that $||s(x) - \bar{x} || = || x - \bar{x} ||$.  We have equality in \eqref{cs-ligning} iff $s(x) - \bar{x}$ and $x - \bar{x}$ are linearly dependent, which is true for  $s(x) - \bar{x} = \pm (x - \bar{x})$.

\end{document}